\documentclass[letterpaper, 10 pt, conference]{ieeeconf} 

\usepackage{amsmath}
\usepackage{amssymb}
\usepackage{dsfont}
\usepackage{bbold}
\usepackage{graphicx}
\usepackage{tikz}
\usetikzlibrary{arrows}

\usepackage{thmtools}
\usepackage{answers}
\usepackage{lipsum}

\renewcommand\thmcontinues[1]{Continued}

\newtheorem{theorem}{Theorem}
\newtheorem{proposition}{Proposition}
\newtheorem{definition}{Definition}

\newtheorem{lemma}{Lemma}
\newtheorem{remark}{Remark}
\newtheorem{example}{Example}

\IEEEoverridecommandlockouts                             

\title{\LARGE \bf
On stability of users equilibria in heterogeneous routing games}

\author{Leonardo Cianfanelli and Giacomo Como\thanks{Department of Mathematical Sciences G.L.~Lagrange, Politecnico di Torino, Corso Duca degli Abruzzi 24, 10129, Torino, Italy.  Email: leonardo.cianfanelli@polito.it, giacomo.como@polito.it. }
\thanks{This reasearch was carried on within the framework of the MIUR-funded {\it Progetto di Eccellenza} of the {\it Dipartimento di Scienze Matematiche G.L.~Lagrange}, CUP: E11G18000350001, and was supported by the {\it Compagnia di San Paolo} through a Starting Grant and a Joint Research Grant.}
\thanks{}
}

\begin{document}

\maketitle
\thispagestyle{empty}
\pagestyle{empty}

\begin{abstract}
The asymptotic behaviour of deterministic logit dynamics in heterogeneous routing games is analyzed. It is proved that in directed multigraphs with parallel routes, and in series composition of such multigraphs, the dynamics admits a globally asymptotically stable fixed point. Moreover, the unique fixed point of the dynamics approaches the set of Wardrop equilibria, as the noise vanishes. The result relies on the fact that the dynamics of aggregate flows is monotone, and its Jacobian is strictly diagonally dominant by columns.
\end{abstract}

\section{INTRODUCTION}





Congestion population games provide a powerful tool to model real situations where users compete for resources.
An application of such games is traffic, where resources are roads, and players using roads create negative externalities because of congestion effects. 
The simplest traffic models always assume homogeneity of players, meaning that each player perceives the same delay over the set of the roads.

A Wardrop equilibrium is a flow distribution such that any player cannot unilaterally decrease his perceived delay by changing route \cite{wardrop}.
Beckmann \emph{et al.} showed that finding a Wardrop equilibrium in the homogeneous setting is equivalent to minimizing a convex potential with linear constraints \cite{beckmann}. 
The existence of such potential has several strong implications: 
firstly, the game always admits at least a Wardrop equilibrium;
secondly, with the additional hypothesis that delays are strictly increasing, the Wardrop equilibrium is unique; 
moreover, the deterministic logit dynamics converges to a perturbation of the Wardrop equilibrium, where the magnitude of perturbations grows with noise level \cite{sandholm2010population}.

The homogeneity assumption simplifies the problem, but it is very restrictive. Heterogeneity, on the one hand, allows the description of many real situations, e.g., when payoffs are a combination of time delays and monetary tolls and users trade off money and time in different ways \cite{cole}, \cite{fleischer2004tolls}, or when users have different information on the state of the roads \cite{wu2017informational}, \cite{thai2016negative}, \cite{acemoglu2018informational}. 
On the other hand, heterogeneous games do not admit a potential \cite{farokhi2013heterogeneous}. Then, existence and uniqueness of the equilibrium in such games are not trivial, as well as convergence of deterministic logit dynamics. 

Existence and uniqueness of Wardrop equilibrium in heterogeneous games have been already investigated in literature. In particular,
existence has been proved to hold, even in a more general setting \cite{schmeidler}. Conversely, unlike the homogeneous case, uniqueness does not hold in general. Uniqueness of aggregate flows (i.e. \emph{essential uniqueness})
has been proved to hold just on a small subset of graphs \cite{Milchtaich}.

The asymptotic behaviour of some nontrivial learning dynamics has been studied also, for instance in stochastic setting, where users, playing a repeated routing game, perceive at each stage a stochastic realization of delays on the edges \cite{meigs2017learning}, \cite{krichene2015convergence}.
To the best of our knowledge, no papers in literature investigate the asymptotic behaviour of deterministic logit dynamics in heterogeneous routing games, which are not potential games.  In most of the applications the authors analyze the efficiency of the equilibria, but do not investigate the convergence to such equilibria \cite{wu2017informational}, \cite{thai2016negative}, \cite{acemoglu2018informational}.
Given these motivations, in this work we focus on this problem, highlighting the role of the graph topology.
In particular, using monotonicity arguments, we show that the dynamics converges to a unique globally asymptotically stable fixed point 
when the graph is simple, i.e. it has parallel routes between the origin and the destination, or is a series composition of such simple graphs. Moreover, we show that the unique fixed point of the dynamics approaches the set of Wardrop equilibria, as the noise vanishes.

The rest of the paper is organized as follows. In Section II we define the model and summarize the state of the art on existence and uniqueness of Wardrop equilibrium. In Section III we present the main result of the paper, which is that the deterministic logit dynamics on simple graphs and series composition of simple graphs converges to a unique globally asymptotically stable fixed point. Then, in Section IV, we present examples and numerical simulations. Finally, in Section V, we summarize the results and discuss future research lines.

\subsection{Notation}
Let $\mathds{R}$, $\mathds{R}_+$ and $\mathds{R}^n$ denote the set of real numbers, non-negative reals, and real-valued vectors of dimension $n$, respectively. Let $|\mathcal{X}|$ denote the cardinality of a countable set $\mathcal{X}$. $\|x\|_{l_1}$ denotes the $l_1$-norm of the vector $x \in \mathds{R}^n$, i.e., $\|x\|_{l_1}=\sum_{i=1}^n |x_i|$. Along the paper, $\mathcal{G}=(\mathcal{V},\mathcal{E})$ will always denote a directed \emph{multigraph}, i.e., a graph admitting parallel edges, even where not explicitly specified. $\mathcal{V}$ and $\mathcal{E}$ denote the node set and the edge set, respectively. Sometimes, along the paper, we will use the word \emph{graph} instead of \emph{multigraph}. However, all of the graphs in this paper are to be read as multigraphs.

\section{MODEL DESCRIPTION}
 Let $\mathcal{G}=(\mathcal{V},\mathcal{E})$ be a directed multigraph with origin $o \in \mathcal{V}$ and destination $d \in \mathcal{V}$.
 We consider several populations of users and denote by $\mathcal{P}$ the set of populations. We assume that every population has the same origin-destination pair $o$-$d$.
 For each population $p$, let $\tau^p \ge 0$ be the throughput which has to go from $o$ to $d$, and let $\tau=\sum_{p \in \mathcal{P}} \tau^p$ be the aggregate throughput. The populations differ each other for the assignment of delay functions over the edge set. 
 Let $d_e^p (x):\mathds{R}_+ \rightarrow \mathds{R}_+ $ be the \emph{delay function} of population $p \in \mathcal{P}$ on the edge $e \in \mathcal{E}$.
In population games, it is assumed that the number of player is large,
the strategy of a single player has negligible effects on payoff functions, and players interact anonymously.
This means that delay functions depend only on the distribution of agents strategies, which in case of traffic are flows on the edges.
 Moreover, we assume \emph{separability}, i.e., each delay function depends on the flow on the edge itself only. The delay functions are assumed to be \emph{continuous} and \emph{non-decreasing}, since edges suffer from congestion. 
 Along the paper we will remark when additional assumption of strictly increasing delays is needed.
 The set of strategies correspond to the set of routes going from $o$ to $d$, and does not depend on populations. 
 Let $\mathcal{R}$ denote the route set, and let $V$, $E$, $P$ and $R$ denote the cardinalities of $\mathcal{V}$, $\mathcal{E}$, $\mathcal{P}$ and $\mathcal{R}$, respectively.

An \emph{admissible route flow distribution} for a population $p \in \mathcal{P}$ is a vector $z^p \in \mathds{R}_+^{R}$ satisfying the throughput constraint, i.e.,
$\sum_{r\in \mathcal{R}} z^p_r = \tau^p$.
For a given route flow vector $z^p$, the (unique) edge flow vector is obtained via
\begin{equation}
\label{incidence}
f^p=Az^p,
\end{equation}
where $A \in \mathds{R}^{E\times R}$ is the edge-route incidence matrix, with entries $A_{er}=1$, if the edge $e$ belongs to the route $r$, or $0$ otherwise.
Let the \emph{aggregate edge flow distribution} and \emph{aggregate route flow distribution} be
\begin{equation}
\label{f_agg}
f=\sum_{p \in \mathcal{P}}f^p, \quad z=\sum_{p \in \mathcal{P}}z^p,
\end{equation}
respectively.
The cost of each route is defined as the sum of delay of edges belonging to the route, i.e.:
\begin{equation}
\label{cost}
c_r^p(z)=\sum_{e \in \mathcal{E}} A_{er} d_e^p(f_e),
\end{equation}
where, given $z^p$ for all the populations, the aggregate edge flow vector $f$ is computed by (\ref{incidence}) and (\ref{f_agg}). 
We remark that, although delay functions on the edges are specific for each population, congestion is not specific, in the sense that delay functions depend on aggregate flows only.

\begin{definition}[Heterogeneous routing game]
An heterogeneous routing game is a triple ($\mathcal{G}$, $\mathcal{P}$, $d$), where $\mathcal{P}$ includes the throughput $\tau^p$ of each population, and $d$ denotes the vector containing the delay functions for each edge and population.
\end{definition}

\medskip

In heterogeneous routing games each player aims at minimizing his cost (\ref{cost}) according to delay functions that are specific for the population he belongs to, given the strategies of all of the other players, i.e., flows. Indeed, since the game is a population game, the strategy of the player itself does not affect the flows, and consequently the costs.

\begin{definition}[Wardrop equilibrium] A Wardrop equilibrium for the heterogeneous routing game is an admissible route flow distribution such that for every population $p \in \mathcal{P}$, and route $r \in \mathcal{R}$
\begin{equation}
\label{wardrop}
z^p_r > 0 \ \Rightarrow \ c^p_r(z) \le c^p_q(z) \quad  \forall q \in \mathcal{R}.
\end{equation}

\end{definition}

\medskip

Thus, at Wardrop equilibrium, no one can unilaterally decrease his cost by changing route, since every route used by a population $p$ has the minimal cost (measured by the population $p$ itself) among all the routes. 

\subsection{Existence and uniqueness of equilibrium}
The main difference between homogeneous and heterogeneous games is that the latter ones do not admit a potential. 
Consequently, existence and uniqueness of equilibrium are not trivial.
However, as already said, the existence of equilibrium in heterogeneous games still holds  \cite{schmeidler}.
Before looking into uniqueness, we introduce the following definition:
\begin{definition}[Essential uniqueness of equilibrium]
The Wardrop equilibrium in an heterogeneous routing game is said to be \emph{essentially unique} 
if the aggregate edge flows are the same for every Wardrop equilibria.
\end{definition}

\medskip

Milchtaich shows in \cite{Milchtaich} that the largest class of graphs where essential uniqueness is guaranteed to hold is \emph{series of nearly parallel} graphs, in the sense that in every routing game on such class of graphs, under the assumption of strictly increasing delay functions, the equilibrium is essentially unique. Conversely, for each graph not belonging to such class, at least a game $(\mathcal{G},\mathcal{P},d)$ such that the equilibrium is not essentially unique exists.\\
The following example illustrates a game where equilibrium is essentially unique, but not unique.
\begin{example}[label=exa:cont]
\label{ex1}
Consider a multigraph $\mathcal{G}=(\mathcal{V},\mathcal{E})$, with
$\mathcal{V}=\{o,d\}$ and $\mathcal{E}=\{e_1,e_2\}$,
where both edges go from origin $o$ to destination $d$. Let $P=2$ and let the delay functions be strictly increasing.
If the assignment of delay functions and throughput admits $\tilde{f}$ such that 
\begin{equation}
\label{continuum}
d_1^1(\tilde{f})=d_2^1(\tau-\tilde{f}), \quad
d_1^2(\tilde{f})=d_2^2(\tau-\tilde{f}),
\end{equation}
all the edge flows satisfying
\begin{equation}
\begin{gathered}
    f_1^1+f_1^2=\tilde{f}, \quad f_2^1+f_2^2=\tau-\tilde{f}, \\
    f_1^1+f_2^1=\tau^1, \quad f_1^2+f_2^2=\tau^2,
\end{gathered}
\label{continuum2}
\end{equation}
are Wardrop equilibria. Thus, even though the equilibrium is unique in terms of aggregate edge flows (essentially unique), this game admits a continuum of equilibria in terms of population edge flows.
\end{example}

\medskip
As previously said, any graph not belonging to the series of nearly parallel class admits at least a game where the equilibrium is not essentially unique. The following example, which is a slight modification of Example 1 in \cite{konishi}, illustrates a game with two essentially different equilibria.
\begin{example}[label=exkon:cont]
\label{ex_konishi}
Consider the multigraph $\mathcal{G}$ in Fig. \ref{kon_img}, and let $P=3$. Let us assign throughputs
\begin{gather*}
\tau^1=1.2, \quad
\tau^2=1, \quad
\tau^3=1,
\end{gather*}
and the following delay functions:
\begin{gather*}
d_1^1(x)=d_2^1(x)=d_4^1(x)=d_6^1(x)=19+x,\\
d_1^2(x)=d_4^2(x)=d_1^3(x)=d_4^3(x)=19+x,\\
d_3^1(x)=d_5^1(x)=d_3^2(x)=100+x,\\
d_6^2(x)=d_2^3(x)=d_5^3(x)=100+x,\\
d_2^2(x)=d_6^3(x)=20x,\\
d_5^2(x)=d_3^3(x)=21+x.
\end{gather*}
Such assignment prevents population $1$ from using $e_3$ and $e_5$, population $2$ from using $e_3$ and $e_6$, and population $3$ from using $e_2$ and $e_5$. This assignment models a problem with more classes of users who ignore the existence of some roads.
Let us denote the routes of $\mathcal{G}$ by $r_1=(e_1,e_2)$, $r_2=(e_1,e_3)$, $r_3=(e_4,e_5)$, and $r_4=(e_4,e_6)$.
This game admits two essentially different equilibria. In both of them, all the users of a population use the same route.
In equilibrium (A), populations $1$, $2$, and $3$ use routes $r_1$, $r_3$, and $r_4$, respectively. In equilibrium (B), the used routes are $r_4$, $r_1$, and $r_2$, respectively.
The aggregate edge flows at the equilibrium (A) are
\begin{equation}
\begin{gathered}
\label{A}
f_1=f_2=1.2\\
f_3=0\\
f_4=2\\
f_5=f_6=1,
\end{gathered}
\end{equation}
while at the equilibrium (B) are
\begin{equation}
\begin{gathered}
\label{B}
f_1=2\\
f_2=f_3=1\\
f_4=f_6=1.2\\
f_5=0.
\end{gathered}
\end{equation}
We omit the proof that (A) and (B) are truly users equilibria, because it is trivial: given the edge flows (\ref{A}) and (\ref{B}), it suffices to observe that the routes used by each population carry the minimum delay among all of the routes.
Since the aggregated flows are different, the equilibrium is not essentially unique.
\end{example}

\medskip

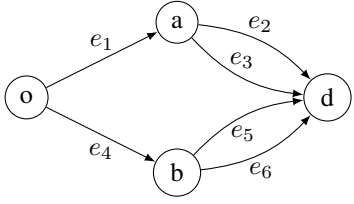
\begin{figure}
\begin{center}
\begin{tikzpicture}
\node[draw, circle] (1) at (0,0) {o};
\node[draw, circle] (2) at (2,1) {a};
\node[draw, circle] (3) at (2,-1) {b};
\node[draw, circle] (4) at (4,0) {d};

\path [->, >=latex]  (1) edge [bend right=0]
node [above] {$e_1$} (2);
\path [->, >=latex]  (1) edge [bend right=0]
node [below] {$e_4$} (3);
\path [->, >=latex]  (2) edge [bend right=20]
node [above] {$e_3$} (4);
\path [->, >=latex]  (2) edge [bend right=-20]
node [above] {$e_2$} (4);
\path [->, >=latex]  (3) edge [bend right=20]
node [below] {$e_6$} (4);
\path [->, >=latex]  (3) edge [bend right=-20]
node [below] {$e_5$} (4);

\end{tikzpicture}
\end{center}
\caption{This graph is not series of nearly parallel graphs. Then, at least a game defined on this graph admitting multiple essentially different equilibria exists. Example \ref{ex_konishi} provides a game with two essentially different equilibria. }
\label{kon_img}
\end{figure}

\section{Logit dynamics and its stability}
Our main contribution in heterogeneous congestion games is on convergence of \emph{deterministic logit dynamics} on a certain class of graphs. This dynamics is standard in the literature of population games (see \emph{logit dynamics} in \cite{sandholm2010population}, or \emph{perturbed best response dynamics} in \cite{como2013stability}), and can be derived from Kurtz's theorem \cite[Ch. 11]{ethier2009markov}, as the mean-field approximation of the intrinsically stochastic noisy best response dynamics, as the number of player grows larger.

Logit dynamics for the heterogeneous routing game reads, for every population $p \in \mathcal{P}$, and every route $r \in \mathcal{R}$:
\begin{equation}
\label{gen_system}
\dot{z}_r^p=\tau^p\cdot\frac{\exp(-\eta\cdot c_r^p(\sum_{q=1}^P z^q))}
{\sum_{s \in R} \exp(-\eta\cdot c_s^p(\sum_{q=1}^p z^q))}-z_r^p,
\end{equation}
where $\eta$ is the inverse of noise level, and the route costs $c_r^p$ depend in general on all the components of the route flow distribution, 
since each edge belongs in general to several routes.

In this section we show that the logit dynamics on simple graphs and series composition of simple graphs converges to one globally asymptotically stable fixed point. Moreover, we show that such fixed point approaches the set of the Wardrop equilibria, as the noise vanishes. 
Let us start by defining those classes of graphs.

\begin{definition}[Simple multigraph]
A multigraph with single origin and destination is said to be \emph{simple} if each edge belongs to one route at most, or, alternatively, routes are parallel.
\end{definition}

\medskip

\begin{definition}[Series composition of multigraphs]
Two multigraphs $\mathcal{G}_1$ and $\mathcal{G}_2$ are said to be connected in series if they have a single common vertex, which is the destination in $\mathcal{G}_1$ and the origin in $\mathcal{G}_2$.
The multigraph $\mathcal{G}=S(\mathcal{G}_1,\mathcal{G}_2)$, where $\mathcal{G}_1$ and $\mathcal{G}_2$ are connected in series, is said to be the \emph{series composition} of $\mathcal{G}_1$ and $\mathcal{G}_2$.
\end{definition}

\medskip

Simple graphs are a subset of nearly parallel graphs. Then, essential uniqueness on series of simple graphs is guaranteed if the delays are strictly increasing.
In Fig. \ref{simple_and_series} two examples of those graphs are provided.

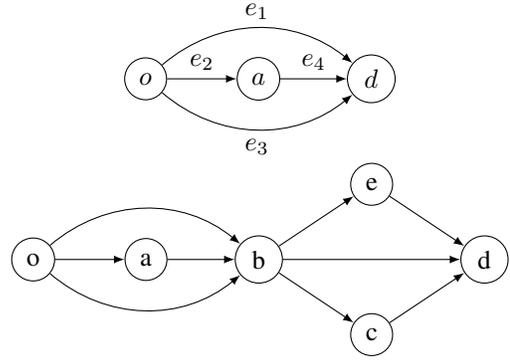
\begin{figure}
\begin{center}
\begin{tikzpicture}
\node[draw, circle] (1) at (1.5,0) {$o$};
\node[draw, circle] (2) at (3.0,0) {$a$};
\node[draw, circle] (3) at (4.5,0) {$d$};

\path [->, >=latex]  (1) edge [bend right=0]
node [above] {$e_2$} (2);
\path [->, >=latex]  (1) edge [bend right=-40]
node [above] {$e_1$} (3);
\path [->, >=latex]  (1) edge [bend right=40]
node [below] {$e_3$} (3);
\path [->, >=latex]  (2) edge [bend right=0]
node [above] {$e_4$} (3);
\end{tikzpicture}

\begin{tikzpicture}
\node[draw, circle] (1) at (0.5,0) {o};
\node[draw, circle] (2) at (2,0) {a};
\node[draw, circle] (3) at (3.5,0) {b};
\node[draw, circle] (4) at (5,-1) {c};
\node[draw, circle] (5) at (5,1) {e};
\node[draw, circle] (6) at (6.5,0) {d};

\path [->, >=latex]  (1) edge [bend right=0]
node [above] {} (2);
\path [->, >=latex]  (1) edge [bend right=-40]
node [above] {} (3);
\path [->, >=latex]  (1) edge [bend right=40]
node [below] {} (3);
\path [->, >=latex]  (2) edge [bend right=0]
node [above] {} (3);
\path [->, >=latex]  (3) edge [bend right=0]
node [above] {} (5);
\path [->, >=latex]  (3) edge [bend right=0]
node [below] {} (4);
\path [->, >=latex]  (3) edge [bend right=0]
node [above] {} (6);
\path [->, >=latex]  (4) edge [bend right=0]
node [below] {} (6);
\path [->, >=latex]  (5) edge [bend right=0]
node [below] {} (6);

\end{tikzpicture}
\end{center}
\caption{A simple graph and a series composition of simple graphs.}
\label{simple_and_series}
\end{figure}

The main result of the paper will be now shown.

\begin{theorem}
[Convergence on series of simple graphs]
\label{teorema}
Let $\mathcal{G}=(\mathcal{V},\mathcal{E})$ be a simple directed multigraph, or a series composition of simple directed multigraphs. Let $\mathcal{P}$ be the set of populations, and $d$ be the vector of non-decreasing delay functions for every edge and population. The continuous logit dynamics \eqref{gen_system} for the game $(\mathcal{G},\mathcal{P},d)$ converges to a unique globally asymptotically stable fixed point. Let $\tilde{f}_e^p(\eta)$ denote such unique fixed point, as function of noise level. As the noise vanishes, $\lim_{n \rightarrow +\infty} \tilde{f}_e^p(\eta)$ approaches the set of Wardrop equilibria.
\end{theorem}

\medskip
\begin{remark}
Theorem \ref{teorema} states the uniqueness of the fixed point under the assumption of non-decreasing delay functions, although on series of simple graphs such assumption does not guarantee either the uniqueness or the essential uniqueness of the equilibrium. In fact, it can be shown that, if the equilibrium is not unique, the dynamics selects the equilibrium where every population randomizes uniformly among its optimal routes. An example will be provided in the next section.
\label{rem_random}
\end{remark}

\medskip
We are going to prove Theorem \ref{teorema} through the following intermediate results. First, we establish Lemma \ref{prop1} on $l_1$-contraction in monotone dynamical systems. Next, we show in Proposition \ref{conv_simple} that the logit dynamics on simple graphs converges to a unique globally stable fixed point. Finally, we will show that, if the dynamics converges in two graphs, the convergence holds for the series composition of them, also.

Let us start with the plan.
\begin{lemma}[$l_1$ contraction] 
\label{prop1}

Let $z \in \mathds{R} ^n$ and $\dot{z}=f(z)$ be an autonomous dynamical system. 
Let $J$ be the Jacobian matrix of $f$. Let $J$ satisfy the following properties:
\begin{enumerate}
\item $J$ is Metzler (i.e $\frac{\partial{f_i}}{\partial{z_j}}\geq 0$, $\forall i \neq j$);
\item $J$ is diagonally dominant, and in particular $\sum_{i}{\frac{\partial{f_i}}{\partial{z_j}}}=-a$, $\forall j$, with $a\ge 0$.
\end{enumerate}
Then $\|\tilde{z}(t)-z(t)\|_{l_1} \leq \|\tilde{z}(0)-z(0)\| _{l_1} e^{-at}$, where $\tilde{z}$ and $z$ are solutions of $\dot{z}=f(z)$.
\end{lemma}

\begin{proof} The proof follows the steps of Lemma 5 in \cite{lovisari2014stability}. See also \cite{como2017resilient}. For details see Appendix A.
\end{proof} 

Instead of studying the dynamics over all the possible simple graphs, 
we are going to investigate just the case of simple graphs composed of $2$ nodes and $E$ parallel edges, going from the origin to the destination.
In such graphs, edge flows and route flows are equivalent, and the cost of each route only depends on the flow on the route itself.
\begin{remark}
\label{rem_simple}
Investigating this case only, instead of all the simple multigraphs, does not determine any loss of generality.
Indeed, all of the removed nodes do not impose any decision to the players. Hence, they can be removed with the prescription that the cost of each route, which is by (\ref{cost}) the sum of delays of all the edges along the route, is assigned to a single edge. 
\end{remark}
\medskip

Thus, logit dynamics on simple graphs reads:
\begin{equation}
\label{system}
\dot{f}_e^p=\tau^p\cdot\frac{\exp(-\eta \cdot d_e^p(\sum_{q=1}^P f_e^q))}{\sum_{j=1}^E \exp(-\eta \cdot d_j^p(\sum_{q=1}^P f_j^q))}-f_e^p.
\end{equation}
This system is composed of $E\times P$ coupled equations describing the evolution in time of flows. 
The main difference between the simple and the general case is that in the simple case the evolution of each edge is affected by the other edge flows only through the normalization, because routes are parallel. 
This is the key point to prove the next proposition.

\begin{proposition}[Convergence on simple graphs]
\label{conv_simple}
Let $\mathcal{G}=(\mathcal{V},\mathcal{E})$ be a simple directed multigraph. Let $\mathcal{P}$ be the set of populations, and $d$ be the vector of non-decreasing delay functions for every edge and population. The continuous logit dynamics \eqref{gen_system} for the game $(\mathcal{G},\mathcal{P},d)$ converges to a unique globally asymptotically stable fixed point. Moreover, as the noise vanishes, $\lim_{\eta \rightarrow +\infty} \tilde{f}_e^p(\eta)$ approaches the set of Wardrop equilibria.
Furthermore, the aggregate flow distribution $f=\sum_{p=1}^P f^p$ converges exponentially in time.
\end{proposition}
\medskip
The proof relies on that the dynamics of the aggregated flows
satisfies the hypotheses of Lemma 1. The assumption of
simple graph is crucial, since it guarantees that the Jacobian
J is Metzler.
For the details, see Appendix B.\\
The next proposition states that, if the dynamics converges on two graphs, it converges on series composition of them also.
\begin{proposition}[Convergence in series of graphs] Let $\mathcal{G}_1$ and $\mathcal{G}_2$ be two directed multigraphs, and
let $\mathcal{G}=S(\mathcal{G}_1,\mathcal{G}_2)$ be their series composition. If the logit dynamics on $\mathcal{G}_1$ and $\mathcal{G}_2$ converges to a globally asymptotically stable fixed point, the logit dynamics on $\mathcal{G}$ converges to a globally asymptotically stable fixed point as well. Morevoer, as the noise vanishes, if both the fixed point of logit dynamics on $\mathcal{G}_1$ and logit dynamics on $\mathcal{G}_2$ approach their respective sets of Wardrop equilibria, the fixed point of logit dynamics on $\mathcal{G}$ approaches the set of Wardrop equilibria on $\mathcal{G}$ as well.
\label{conv_ser}
\end{proposition}
\medskip
For the proof, see Appendix C.
We can now conclude the proof of Theorem \ref{prop1}.

\begin{proof}[of Theorem \ref{prop1}]
By Proposition \ref{conv_simple}, the logit dynamics on simple graphs converges to a unique globally asymptotically stable fixed point. Then, by applying Proposition \ref{conv_ser} recursively, logit dynamics on series compositions of an arbitrary number of simple graphs converges to a unique globally asymptotically stable fixed point. Moreover, by Propositions \ref{conv_simple}, as the noise vanishes, the fixed point of dynamics approaches the set of the Wardrop equilibria. Then, by Proposition \ref{conv_ser}, the same holds for series compositions of simple graphs.
\end{proof}

\section{Examples and numerical simulations}
In this section three examples are shown. The first example shows that, on simple graphs, even if the equilibrium is not unique, the logit dynamics selects the equilibrium where each population randomizes uniformly among his optimal routes. The second example illustrates the role played by the noise. Finally, the last example shows that when the equilibrium is not essentially unique, and the noise is sufficiently small, the dynamics selects one equilibrium based on the initial conditions.

\begin{example}[continues=exa:cont]
\label{ex2}
Consider a multigraph $\mathcal{G}$, with $2$ nodes and $2$ parallel edges from origin to destination, and the following assignment:
\begin{gather*}
    d_1^1(f_1)=f_1+1, \quad d_1^2(f_1)=2f_1, \quad
    \tau^1=1; \\
    d_2^1(f_2)=2f_2, \quad d_2^2(f_2)=f_2+1; \quad \tau^2=1.
\end{gather*}
It is easy to check that this assignment satisfies conditions (\ref{continuum}) and (\ref{continuum2}). Then, the game admits a continuum of Wardrop equilibria. In particular, any edge flow distribution satisfying
\begin{equation*}
    f_1^1=f_1^1, f_2^1=1-f_1^1, f_1^2=1-f_1^1, f_2^2=f_1^1,
\end{equation*}
is a Wardrop equibrium. 
Nevertheless, numerical simulations in Fig. \ref{non_uniqueness} show that the dynamics selects the equilibrium where both the populations split their throughput among both routes, as claimed in Remark \ref{rem_random}.
\end{example}
\medskip

\begin{example}
Consider the simple multigraph in Fig. \ref{simple_and_series}. 
Let 
\begin{gather*}
d_1^1(f_1)=f_1+1, \quad d_1^2(f_1)=(f_1)^2+1; \\
d_2^1(f_2)=\frac{f_2}{2}+2, \quad d_2^2(f_2)=f_2+2; \\ d_3^1(f_3)=(f_3)^2+1, \quad d_3^2(f_3)=f_3+2; \\
d_4^1(f_4)=\frac{f_4}{2}, \quad d_4^2(f_4)=f_4;\\
\tau^1=5, \quad \tau^2=5.
\end{gather*}
In Fig. \ref{img_ex3}, numerical simulations of logit dynamics, as the noise decreases, are shown. When $\eta=0$, the dynamics randomizes among all the routes, since the noise diverges. As $\eta$ grows larger, the trajectories converge closer to the Wardrop equilibrium. We also show that, starting with two different initial conditions, the aggregate flows converge each other at least exponentially, as stated in Proposition \ref{conv_simple}. 
\label{ex3}
\end{example}

\medskip

\begin{example}[continues=exkon:cont]
Consider again the multigraph $\mathcal{G}$ in Fig. \ref{kon_img}. In the first part of the example a game with two essentially different equilibria has been provided. Let us now focus on the logit dynamics for. Of course, when $\eta=0$, at the equilibrium, every population randomizes among all the routes. However, Fig. \ref{sim_konishi} shows that, if $\eta$ is sufficiently large, the trajectories converge close to one of the equilibria, depending on initial conditions. Thus, the logit dynamics on this graph has a bifurcation point: if the noise is sufficiently large, the dynamics admits a globally asymptotically stable fixed point, while, as the noise decreases, the dynamics admits two asymptotically stable fixed points.
\end{example}

\begin{figure}
\begin{center}
\includegraphics[width=5cm]{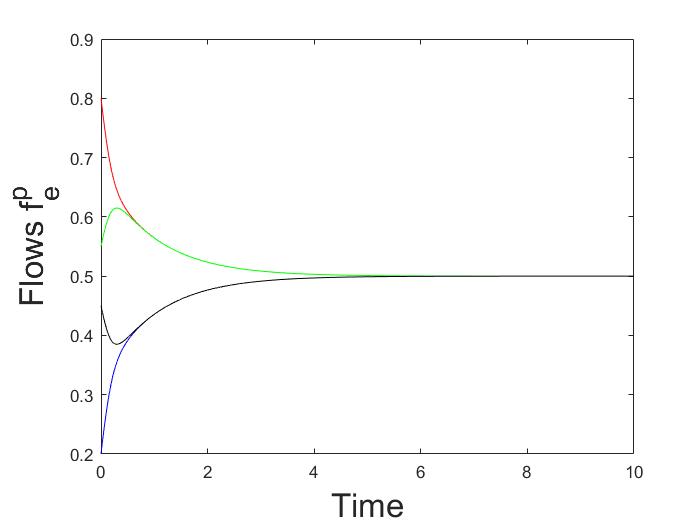}
\end{center}
\caption{A logit dynamics trajectory for Example \ref{ex2}, with $\eta=5$. Different curves correspond to different elements $f_e^p$.}
\label{non_uniqueness}
\end{figure}

\begin{figure}
\begin{center}
\includegraphics[width=4cm]{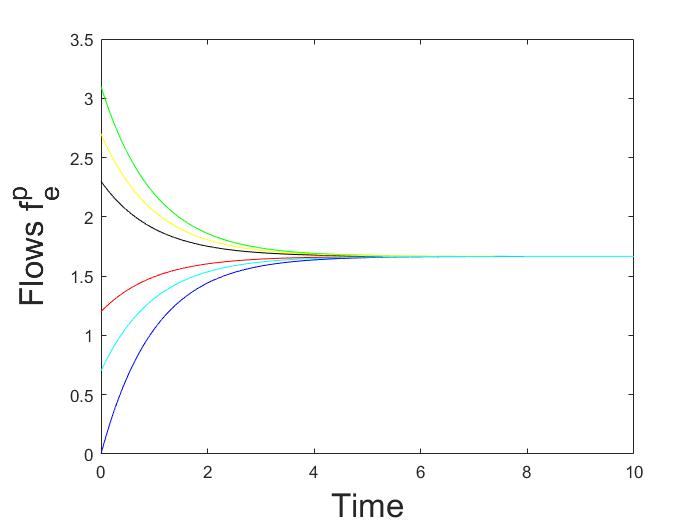}
\includegraphics[width=4cm]{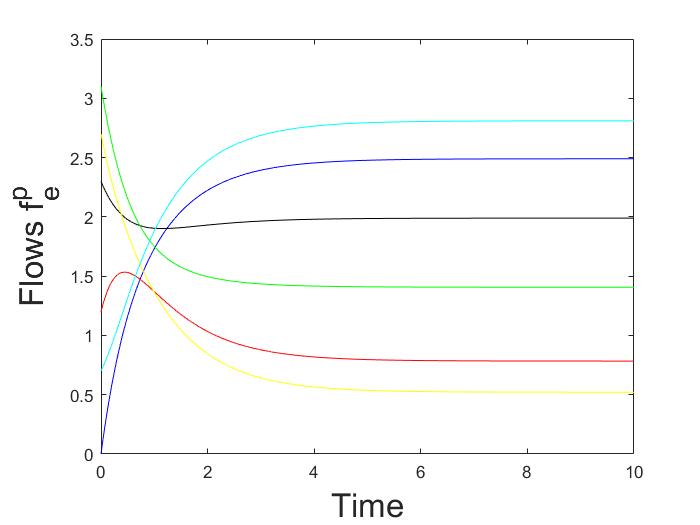}
\includegraphics[width=4cm]{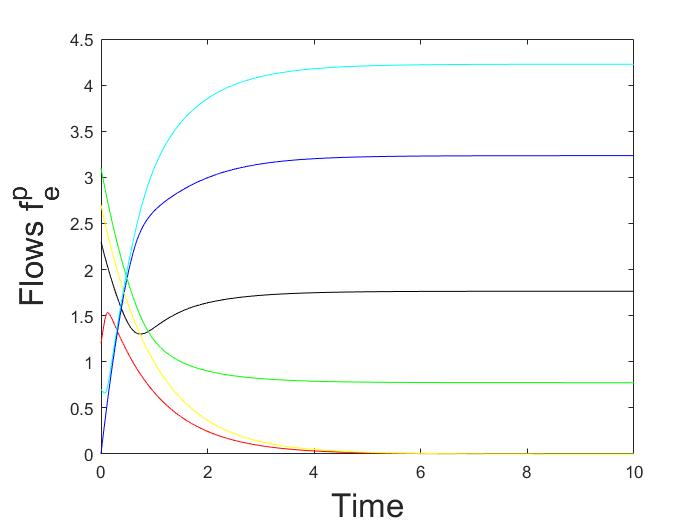}
\includegraphics[width=4cm]{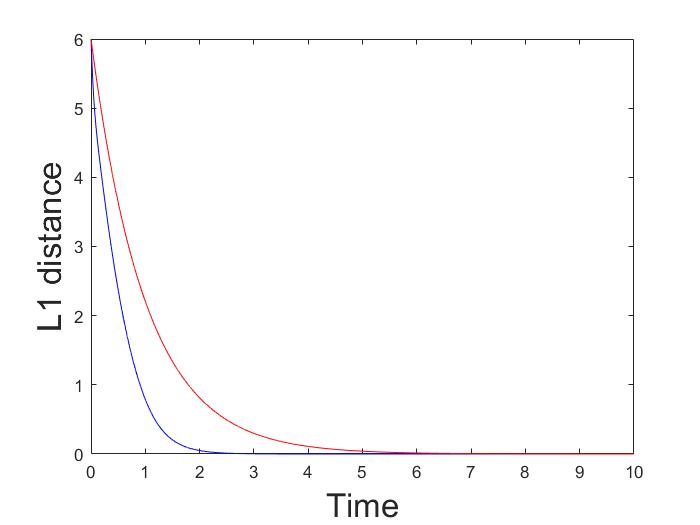}
\end{center}
\caption{Trajectories of logit dynamics for Example \ref{ex3} \emph{Top left}: $\eta=0$. \emph{Top right}: $\eta=0.2$. \emph{Bottom left}: $\eta=2$. \emph{Bottom right}: Let $g$ and $h$ be two trajectories for the evolution of the aggregate flows with different initial conditions, and let $\eta=2$; the blue curve is $\|g(t)-h(t)\|_{l_1}$, while the red one is $\|g(0)-h(0)\|_{l_1}e^{-t}$; as proved in Proposition \ref{conv_simple}, the blue one is below the red one.}
\label{img_ex3}
\end{figure}

\begin{figure}
\begin{center}
\includegraphics[width=5cm]{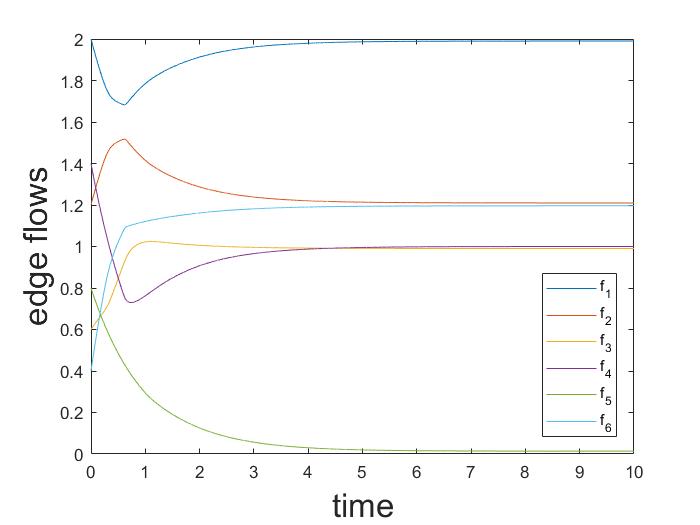}
\includegraphics[width=5cm]{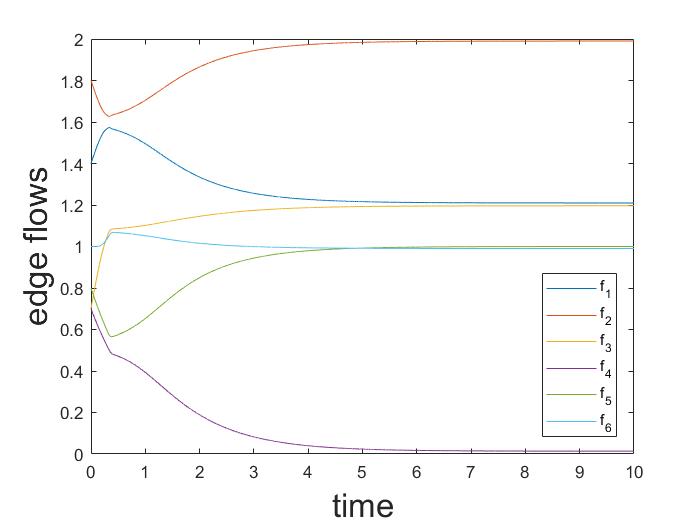}
\end{center}
\caption{Trajectories of logit dynamics on the graph in Fig. \ref{kon_img}, with $\eta=10$. The two graphs show two trajectories starting from different initial conditions. Different colors correspond to different components of the aggregate edge flows. We observe that the asymptotic behaviour of dynamics depends on initial condition. 
}
\label{sim_konishi}
\end{figure}

\section{Conclusions and future works}
In this work the asymptotic behaviour of deterministic logit dynamics in heterogeneous routing games is analyzed, where by heterogeneous we mean that several populations of users, which differ for delay functions, are admitted. 
We prove that on simple graphs, and series composition of simple graphs, the dynamics converges to a globally asymptotically stable fixed point. The proof follows from monotonicity properties of the dynamics.

We are aware that series composition of simple graphs is not the largest class of graphs where uniqueness of equilibrium holds. Milchtaich proved in \cite{Milchtaich} that the largest class where essential uniqueness holds, regardless of the assignment of delay functions, is series of nearly parallel graphs. Investigating the behaviour of the dynamics on nearly parallel graphs other than simple is one of the future research lines. However, the arguments used along this paper cannot be applied for graphs where routes are not parallel.

Furthermore, Example \ref{ex1} suggests that, if the equilibrium is not essentially unique and the noise is sufficiently small, the dynamics converges to one among the equilibria. Since the graphs where essential uniqueness holds are a negligible fraction of real graphs, a further interesting question is whether the dynamics always converge to one among the equilibria in case of multiple essentially different equilibria. Finally, in order to extend these results to more realistic applications, another future research line will be the characterization of the dynamics in case of multiple origin-destination pairs. Unfortunately, even for simple graphs, the dynamics is not monotone. Thus, different techniques should be used.  

\section*{Appendix A}
 
\section*{Proof of Lemma 1}
By definition of the $l_1$-norm and the linearity of the derivative, we get 
\begin{equation} \label{eq1}
\begin{split}
\frac{d}{dt}\|\tilde{z}-z\|_{l_1}   = &\frac{d}{dt} \sum_{i} |\tilde{z_i}-z_i|=\sum_{i}\frac{d}{dt}|\tilde{z_i}-z_i|\\
                                    = & \sum_{i}\text{sign}(\tilde{z_i}-z_i)(\dot{\tilde{z_i}}-\dot{z_i})\\
                                   = & \sum_{i}\text{sign}(\tilde{z_i}-z_i)(f_i(\tilde{z})-f_i(z))\\
																	= &  \sum_{i}\text{sign}(\tilde{z_i}-z_i)(f_i(z+h)-f_i(z)),
\end{split}																	
\end{equation} 
where $\tilde{z}=z+h$. From 
\begin{equation*}
\begin{split}
f_i(z+h)-f_i(z)=& \int_0 ^1 \frac{df_i(z+\tau h)}{d\tau} d\tau \\
=&\int_0 ^1 \nabla f_i(z+\tau h) \cdot h d\tau\\
=&\int_0 ^1 \sum_{j}{\frac{\partial{f_i}}{\partial{z_j}}}h_j d\tau,
\end{split}	
\end{equation*}
eq. (\ref{eq1}) is equal to 
\begin{equation*} \label{eq2}
\int_0 ^1 \sum_{i}\text{sign}(h_i) \sum_{j}{\frac{\partial{f_i}}{\partial{z_j}}}h_j d\tau.
\end{equation*}
Because $\text{sign}(h_i) \in \{-1,0,1\}$, it follows
\begin{equation} \label{eq5}
\begin{split}
&
\sum_{i}\text{sign}(h_i)  \sum_{j}{\frac{\partial{f_i}}{\partial{z_j}}} h_j =\\
& + \sum_{i \in I_{+}}(\sum_{j \in I_{+}}{\frac{\partial{f_i}}{\partial{z_j}} h_j}+\sum_{j \in I_{-}}{\frac{\partial{f_i}}{\partial{z_j}} h_j})\\  
& -\sum_{i \in I_{-}}(\sum_{j \in I_{+}}{\frac{\partial{f_i}}{\partial{z_j}} h_j}+\sum_{j \in I_{-}}{\frac{\partial{f_i}}{\partial{z_j}} h_j}).
\end{split}
\end{equation}
If we consider just the first part of the second term of the last equation 
we get
\begin{equation*}
+\sum_{i \in I_{+}} \sum_{j \in I_{+}} {\frac{\partial{f_i}}{\partial{z_j}}h_j} + 
\sum_{i \in I_{+}} \sum_{j \in I_{-}} {\frac{\partial{f_i}}{\partial{z_j}} h_j},
\end{equation*}
where 
\begin{equation*} \label{eq3}
+\sum_{i \in I_{+}} \sum_{j \in I_{+}} {\frac{\partial{f_i}}{\partial{z_j}}h_j} = 
+ \sum_{j \in I_{+}} h_j \sum_{i \in I_{+}} {\frac{\partial{f_i}}{\partial{z_j}}} \leq 
-a\sum_{j \in I_{+}} |h_j| \le 0
\end{equation*}
because for each $j \in I_{-}$ the sum $\sum_{i \in I_{+}} {\frac{\partial{f_i}}{\partial{z_j}}}$ contains for sure the element on the diagonal, 
but may not contain all the elements of the $j^{th}$ column out of the diagonal.\\ 
At the same time
\begin{equation*} \label{eq4}
+ \sum_{i \in I_{+}} \sum_{j \in I_{-}} {\frac{\partial{f_i}}{\partial{z_j}} h_j} \leq 0,
\end{equation*}
because $h_j=-|h_j|$ for each $j \in I_{-}$ and each term $\frac{\partial{f_i}}{\partial{z_j}}$ do not belong to the diagonal and it is non-negative by hypothesis.\\
Similarly we can operate on the second part of (\ref{eq5}) and we obtain
\begin{equation*}
\begin{split}
\frac{d}{dt}\|\tilde{z}-z\|_{l_1} & =\int_0 ^1 \sum_{i}\text{sign}(h_i) \sum_{j}{\frac{\partial{f_i}}{\partial{z_j}}}h_j d\tau \\
&\le - a \int_0 ^1 (\sum_{j \in I_{+}} |h_j| + \sum_{j \in I_{-}} |h_j|) d\tau \\
                                  & = -a\|\tilde{z}-z\|_{l_1} 
\end{split}																	
\end{equation*}

\section*{Appendix B}
 
\section*{Proof of Proposition \ref{conv_simple}}
Let us consider the evolution of aggregate flows $f_e=\sum_{p=1}^Pf_e^p$: 
\begin{equation}
\label{autonomous}
\dot{f_e}=\sum_{p=1}^P[\tau^p\cdot\frac{\exp(-\eta \cdot d_e^p(f_e))}{\sum_{j=1}^E \exp(-\eta \cdot d_j^p(f_j))}]-f_e.
\end{equation}
The system above has some interesting properties.
Firstly, it is autonomous in the aggregate flows.
Secondly, it is monotone, i.e. the non-diagonal elements of the Jacobian are non-negative: 
indeed, $\forall i \neq j$, $\dot{f}_i$ depend on $f_j$ only by the normalization.
Moreover
\begin{equation}
\label{25}
\sum_{e=1}^E\dot{f_e}=\tau-\sum_{e=1}^Ef_e.
\end{equation}
From (\ref{25}) it follows that the Jacobian is strictly diagonally dominant by columns, 
in particular $\sum_{i=1}^E J_{ij}=-1$ $\forall j$.\\
Then, (\ref{autonomous}) satisfies the hypotheses of Lemma \ref{prop1} with $a=1$,
and we can conclude that it admits one globally exponentially stable fixed point. \\
It still remains to prove that, since aggregate flows converge, 
also the edge flows of every population converge to a unique globally asymptotically stable fixed point.
To this end, let us write (\ref{system}) as:
\begin{equation}
\label{26}
\dot{f_e^p}=F_e^p(f)-f_e^p,
\end{equation}
where $f \in \mathds{R}_+^E$ is the aggregate edge flow distribution, which evolves according to (\ref{autonomous}), independently of population flows $f^p$, and
\begin{equation}
F_e^p(f)=\tau^p\cdot\frac{\exp(-\eta \cdot d_e^p(f_e))}{\sum_{j=1}^E \exp(-\eta \cdot d_j^p(f_j))}.
\label{Fep}
\end{equation}
We now show that, since every aggregated edge flow $f_e$ converges, every population edge flow $f_e^p$ converges as well. Observe that $F_e^p(f)$ is continuous for every edge and population.
Let $\tilde{f}= \lim_{t\rightarrow +\infty} f(t)$ denote the fixed point of aggregate flow distribution. 
By continuity of $F_e^p(f)$, and convergence of (\ref{autonomous}), we get that $\forall \epsilon>0$, $\exists T>0$ such that
\begin{equation}
\label{27}
|F_e^p(f(t))-F_e^p(\tilde{f})|<\epsilon \quad \forall t>T, \forall e \in \mathcal{E}, \forall p \in \mathcal{P}.
\end{equation}
It follows:
\begin{equation}
\label{dis}
F_e^p(\tilde{f})-\epsilon-f_e^p < F_e^p(f(t))-f_e^p < F_e^p(\tilde{f})+\epsilon-f_e^p \quad \forall t>T.
\end{equation}

Since evolution of flows does not depend explicitly on time, we can, without loss of generality, translate the axis of time by $-T$, so that (\ref{dis}) holds $\forall t> 0$. Moreover, by continuity of delay functions, in $t=0$,
\begin{equation}
\label{dis0}
F_e^p(\tilde{f})-\epsilon-f_e^p \le F_e^p(f(0))-f_e^p \le F_e^p(\tilde{f})+\epsilon-f_e^p.
\end{equation}

It is easy to prove that:
\begin{equation}
\begin{gathered}
 {f}_e^p(t) \ge (f_e^p(0)-F_e^p(\tilde{f})+\epsilon)e^{-t}+F_e^p(\tilde{f})-\epsilon;\\
 {f}_e^p(t) \le (f_e^p(0)-F_e^p(\tilde{f})-\epsilon)e^{-t}+F_e^p(\tilde{f})+\epsilon,
\end{gathered}
 \label{gath2}
\end{equation}
where the right terms in (\ref{gath2}) are solutions of  
\begin{equation}
\begin{gathered}
 \dot{f}_e^p(t) = F_e^p(\tilde{f})-\epsilon-f_e^p;\\
 \dot{f}_e^p(t) = F_e^p(\tilde{f})+\epsilon-f_e^p,
\end{gathered}
\label{gath3}
\end{equation}
respectively,
with initial condition $f_e^p(0)$.
Indeed, let us assume that the first inequality in (\ref{gath2}) does not hold, i.e., it exists a time $t_2 > 0$ such that
\begin{equation}
f_e^p(t_2)<(f_e^p(0)-F_e^p(\tilde{f})+\epsilon)e^{-t_2}+F_e^p(\tilde{f})-\epsilon.
\end{equation}
Then, it must exist a time $t_1$ such that $0 \le t_1<t_2$,
\begin{equation}
f_e^p(t_1)=(f_e^p(0)-F_e^p(\tilde{f})+\epsilon)e^{-t_1}+F_e^p(\tilde{f})-\epsilon,
\end{equation}
and
\begin{equation}
\dot{f}_e^p(t_1)=F_e^p(f(t_1))-f_e^p(t_1)<F_e^p(\tilde{f})-\epsilon-f_e^p(t_1),
\end{equation}
which is impossible by \eqref{dis0}.
A similar argument can be used to prove the second inequality in (\ref{gath2}).
Then, from (\ref{gath2}), it follows:
\begin{equation}
F_e^p(\tilde{f})+\epsilon \ge \lim_{t\rightarrow +\infty} f_e^p(t) \ge F_e^p(\tilde{f})-\epsilon,
\end{equation}
and since we can choose $\epsilon$ arbitrarily small,
\begin{equation}
\lim_{t\rightarrow +\infty} f_e^p(t) = F_e^p(\tilde{f}) \quad \forall e \in \mathcal{E}, \forall p \in \mathcal{P}.
\end{equation}
Hence, $f_e^p(t)$ converges to a unique fixed point for each initial condition. \\
Let $\tilde{f}_e^p(\eta)$ denote the unique fixed point as function of noise level. Let $\eta_n$ be an infinite sequence, such that $\lim_{\eta \rightarrow +\infty} \eta_n = +\infty$. Since $\tilde{f}_e^p(\eta_n)$ is compact, $\eta_n$ admits a subsequence $\eta_{n_k}$ such that $\tilde{f}_e^p(\eta_{n_k})$ converges. Let $(\tilde{f}_e^p)^*$ denote such limit.

From (\ref{Fep}), it follows:
\begin{equation}
    \lim_{\eta \rightarrow +\infty} F_e^p(\eta, \tilde{f}^*) =  \begin{cases}
        \frac{\tau^p}{|\mathcal{E}_{opt}^p(\tilde{f}^*)|}, \quad \text{if $e \in \mathcal{E}_{opt}^p(\tilde{f}^*)$,}
        \\
        0, \hspace{1.5cm} \text{otherwise},
        \end{cases}
        \label{cases}
\end{equation}
where $\mathcal{E}_{opt}^p(\tilde{f}^*)$ denotes the set of edges that have minimal cost, given $\tilde{f}^*$, for population $p$, and we recall that the equivalence between routes and edges holds by assumption of simple graph and Remark \ref{rem_simple}. Then, since (\ref{gen_system}) admits a unique fixed point, from (\ref{26}), we get:
\begin{equation}
    (\tilde{f}_e^p)^*=\lim_{\eta \rightarrow +\infty} F_e^p(\eta, \tilde{f}^*).
    \label{flimit}
\end{equation}
Equations (\ref{cases}) and (\ref{flimit}) imply that, for every edge $e \in \mathcal{E}$, and every population $p \in \mathcal{P}$:
\begin{equation}
(\tilde{f}^p_e)^* > 0 \ \Rightarrow \ d^p_e(\tilde{f}^*) \le d^p_l(\tilde{f}^*), \quad  \forall l \in \mathcal{E}, 
\end{equation}
that is equivalent to the definition of the Wardrop equilibrium \eqref{wardrop} in simple graphs with $2$ nodes and $E$ parallel edges. Moreover, we also observe that, even if the game admits multiple equilibria (with same aggregated flows, since the graph is simple) $(\tilde{f}^p_e)^*$ represents a specific equilibrium. Indeed, from (\ref{cases}) and (\ref{flimit}), we observe that in $(\tilde{f}^p_e)^*$ each population randomizes uniformly among its optimal routes.
Then, every subsequence $\tilde{f}_e^p(\eta_{n_k})$ converges to the Wardrop equilibrium where each population randomizes uniformly among its optimal routes. Consequently, the same holds for $\tilde{f}_e^p(\eta_n)$ and $\tilde{f}_e^p(\eta)$.

\section*{Appendix C}
\section*{Proof of Proposition \ref{conv_ser}}
Every route in $\mathcal{G}$ is composed of a route in $\mathcal{G}_1$ concatenated with a route in $\mathcal{G}_2$. 
Given this motivation, let us denote every route of $\mathcal{G}$ by a double index, 
where the first one refers to routes in $\mathcal{G}_1$, and the second one to routes in $\mathcal{G}_2$.
Then, (\ref{gen_system}) reads:
\begin{equation}
\dot{z}_{ij}^p=\tau^p\cdot\frac{\exp(-\eta\cdot c_{ij}^p(\sum_{q=1}^P z^q))}
{\sum_{\substack{l \in \mathcal{R}_1 \\ m\in \mathcal{R}_2}} \exp(-\eta\cdot c_{lm}^p(\sum_{q=1}^p z^q))}-z_{ij}^p.
\end{equation}
Since cost functions are additive, we have $c_{ij}(z)=c_i(z)+c_j(z)$. Then:
\begin{equation}
\begin{split}
\dot{z}_{ij}^p=\tau^p\cdot & 
\frac{\exp(-\eta\cdot c_{i}^p(\sum_{q=1}^P z^q))}{\sum_{l \in \mathcal{R}_1} \exp(-\eta\cdot c_{l}^p(\sum_{q=1}^p z^q))} \\  \cdot &
\frac{\exp(-\eta\cdot c_{j}^p(\sum_{q=1}^P z^q))}{\sum_{m \in \mathcal{R}_2} \exp(-\eta\cdot c_{m}^p(\sum_{q=1}^p z^q))}-z_{ij}^p.
\label{series}
\end{split}
\end{equation}
Obviously, we have 
\begin{equation}
\begin{gathered}
z_i^p=\sum_{j \in \mathcal{R}_2}z_{ij}^p \quad \forall i \in \mathcal{R}_1,\\
z_j^p=\sum_{i \in \mathcal{R}_1}z_{ij}^p \quad \forall j \in \mathcal{R}_2.
\end{gathered}
\label{gath}
\end{equation}
Then, by summing over the index $j$ (or $i$) in \eqref{series}, we get the following dynamics for $z_i$ (or $z_j$):
\begin{equation}
\label{series_i}
\dot{z}_i^p=\tau^p\cdot\frac{\exp(-\eta\cdot c_i^p(\sum_{q=1}^P z^q))}{\sum_{l \in \mathcal{R}_1} \exp(-\eta\cdot c_l^p(\sum_{q=1}^p z^q))}-z_i^p.
\end{equation}
We observe that (\ref{series_i}) is equivalent to (\ref{gen_system}), meaning that dynamics on $\mathcal{G}_1$ (and $\mathcal{G}_2$) induced by logit dynamics on $\mathcal{G}$ is equivalent to logit dynamics on $\mathcal{G}_1$ (and $\mathcal{G}_2$) separately.
Let us assume that logit dynamics on both $\mathcal{G}_1$ and $\mathcal{G}_2$ converges to a globally asymptotically stable fixed point. Following the steps of the proof of Proposition \ref{prop1}, (\ref{series}) can be written as
\begin{equation}
\dot{z}_{ij}^p=F_{ij}^p(z_{\mathcal{G}_1},z_{\mathcal{G}_2})-z_{ij}^p,
\label{z1z2}
\end{equation}
where $z_{\mathcal{G}_1}$ and $z_{\mathcal{G}_2}$ denote the route flow distributions on $\mathcal{G}_1$ and $\mathcal{G}_2$, respectively.
Since $z_{\mathcal{G}_i}$ and $z_{\mathcal{G}_j}$ converge for every initial condition by assumption, the convergence of \eqref{z1z2} easily follows.\\
Moreover, since the route choice on $\mathcal{G}$ is decoupled in two independent route choices on $\mathcal{G}_1$ and $\mathcal{G}_2$, the Wardrop equilibria on $\mathcal{G}$ are the compositions of the equilibria on $\mathcal{G}_1$ and $\mathcal{G}_2$. Thus, since logit dynamics on $\mathcal{G}$, and logit dynamics on $\mathcal{G}_1$ and $\mathcal{G}_2$ separately, converge to same route flow distributions $z_{\mathcal{G}_1}$ and $z_{\mathcal{G}_2}$, and since the fixed point of logit dynamics on $\mathcal{G}_1$ and $\mathcal{G}_2$ approach the set of the equilibria as the noise vanishes, we conclude that the fixed point of logit dynamics on $\mathcal{G}$ approaches the set of Wardrop equilibria, as the noise vanishes.

\bibliographystyle{IEEEtran}
\bibliography{bibliografia}
\end{document}